\documentclass[11pt]{article}

\usepackage[utf8]{inputenc}
\usepackage[T1]{fontenc}
\usepackage{amsmath,amsthm}
\usepackage{newtxtext}
\usepackage[amssymbols]{newtxmath} 
\usepackage{graphicx}
\usepackage{xcolor}
\usepackage{booktabs}
\usepackage{algorithm}
\usepackage{algpseudocode}
\usepackage[colorlinks=true,linkcolor=blue!70!black,citecolor=blue!70!black,urlcolor=blue!70!black]{hyperref}
\usepackage{cite}
\usepackage[margin=1in]{geometry}
\usepackage{microtype}
\usepackage{caption}
\usepackage{bm}
\usepackage{multirow}
\newtheorem{theorem}{Theorem}[section]
\newtheorem{lemma}[theorem]{Lemma}
\newtheorem{proposition}[theorem]{Proposition}
\newtheorem{corollary}[theorem]{Corollary}
\newtheorem{remark}{Remark}[section]

\graphicspath{{figures/}}

\title{\LARGE Target-Rate Least-Squares Power Allocation\\over Parallel Channels}
\author{Bhaskar Krishnamachari\\[4pt]
{\normalsize Ming Hsieh Department of Electrical and Computer Engineering}\\
{\normalsize University of Southern California}\\[2pt]
{\normalsize \texttt{bkrishna@usc.edu}}}
\date{\normalsize March 2026}

\begin{document}
\maketitle

\begin{abstract}
We study power allocation over $N$ parallel Gaussian channels, such as OFDM subcarriers, when each channel has a desired target spectral efficiency.
Given channel gain-to-noise coefficients $a_i>0$ and per-channel targets $T_i\ge 0$, we minimize the total squared rate deviation $\sum_{i=1}^{N}(\log_2(1+a_iP_i)-T_i)^2$ subject to a sum-power constraint $\sum_i P_i \le P_{\mathrm{tot}}$ and nonnegativity $P_i \ge 0$.
We prove that the optimal allocation never overshoots any target and may leave power unused when all targets are jointly feasible, a structure fundamentally different from classical waterfilling.
Using the KKT conditions, we derive a per-channel closed-form solution in terms of the Lambert~W function on the active set and reduce the remaining computation to a one-dimensional monotone bisection for the dual variable.
The resulting algorithm runs in $O(N\log(1/\varepsilon))$ time and achieves up to 1{,}890$\times$ speedup over general-purpose numerical solvers at $N=1{,}024$ channels.
Numerical experiments over Rayleigh fading channels confirm that the closed-form solution matches numerical optimization to machine precision and demonstrate superior target-tracking performance compared to waterfilling, uniform allocation, and proportional fairness across a range of operating conditions.
\end{abstract}

\noindent\textbf{Keywords:} Power allocation, OFDM, target rate, Lambert~W function, convex optimization, waterfilling, KKT conditions.

\bigskip

\section{Introduction}
\label{sec:intro}

Orthogonal frequency-division multiplexing (OFDM) and its multiple-access variant OFDMA convert frequency-selective wireless channels into collections of parallel narrowband subchannels, enabling flexible per-subcarrier power and rate adaptation~\cite{tse2005fundamentals,goldsmithvaraiya1997}.
The classical approach to power allocation over these parallel channels maximizes the sum rate under a total power budget, yielding the well-known waterfilling policy~\cite{coverandthomas2006,tse2005fundamentals}.
Waterfilling allocates more power to stronger channels and always exhausts the entire budget, a strategy that is optimal when the sole objective is aggregate throughput.

However, many practical wireless systems operate with \emph{target rates} rather than unbounded throughput objectives.
Per-user quality-of-service (QoS) requirements, application-layer demands, scheduler-assigned operating points, and discrete modulation and coding scheme (MCS) levels all impose desired rates on individual subchannels~\cite{wong1999multiuserofdm,huangberry2009chapter}.
Existing formulations address targets primarily through hard constraints, either as minimum rate floors (rate-adaptive allocation) or as exact requirements to be met with minimum power (margin-adaptive allocation)~\cite{chow1995dmt,sadr2009survey}.
Hard constraints, however, are poorly suited to scenarios where targets are \emph{negotiable}: soft QoS preferences, MCS granularity that prevents exact rate satisfaction, or differentiable optimization pipelines that benefit from smooth penalty functions.

We address this gap by studying a \emph{soft target-rate tracking} formulation.
Rather than enforcing rate constraints, we minimize the sum of squared deviations between achieved rates and channel-specific targets across all subchannels.
This least-squares objective penalizes shortfalls smoothly and leads to an allocation that balances power across channels to approach each target as closely as the budget allows.

To the best of our knowledge, the specific combination of (i)~a quadratic rate-deviation objective, (ii)~parallel Gaussian channels with per-channel targets, (iii)~a sum-power constraint, and (iv)~a Lambert~W closed-form solution has not been studied previously.
The key contributions of this work are:
\begin{enumerate}
\item We formulate the target-rate least-squares power allocation problem over $N$ parallel Gaussian channels with per-channel targets and a sum-power constraint (Section~\ref{sec:model}).
\item We prove that optimal rates never exceed their targets (no-overshoot property), use this to establish convexity of the problem on its natural domain, and characterize two operating regimes depending on whether the power budget exceeds the sum of per-channel caps (Section~\ref{sec:properties}).
\item We derive a per-channel closed-form solution via the KKT conditions and the Lambert~W function, reducing the problem to a one-dimensional monotone root-finding problem for the dual variable (Section~\ref{sec:kkt}).
\item We present a bisection algorithm with guaranteed $O(N\log(1/\varepsilon))$ convergence and discuss numerical stability considerations (Section~\ref{sec:algorithm}).
\item We validate the closed-form solution against numerical optimization, compare allocation behavior with classical waterfilling, and demonstrate scalability and statistical performance over fading channels (Section~\ref{sec:results}).
\end{enumerate}

The remainder of this paper is organized as follows.
Section~\ref{sec:related} reviews related work.
Section~\ref{sec:model} presents the system model and problem formulation.
Section~\ref{sec:properties} establishes structural properties and convexity of the problem on its natural domain.
Section~\ref{sec:kkt} derives the Lambert~W closed form via KKT conditions.
Section~\ref{sec:algorithm} describes the bisection algorithm and its convergence.
Section~\ref{sec:results} presents numerical experiments.
Section~\ref{sec:discussion} discusses extensions and connections to other formulations.
Section~\ref{sec:conclusion} concludes the paper.

\section{Related Work}
\label{sec:related}

\subsection{Classical Waterfilling}

The waterfilling solution for sum-rate maximization over parallel Gaussian channels is a foundational result in information theory~\cite{coverandthomas2006,tse2005fundamentals}.
Goldsmith and Varaiya~\cite{goldsmithvaraiya1997} established the capacity of fading channels with channel side information, extending waterfilling to time-varying settings.
The key property of waterfilling is that it always expends the entire power budget and favors channels with higher gains.
Our formulation departs from this philosophy: the objective is not throughput maximization but target tracking, which can result in unused power.

\subsection{QoS-Constrained and Margin-Adaptive Allocation}

Wong~et~al.~\cite{wong1999multiuserofdm} introduced multiuser OFDM with adaptive subcarrier, bit, and power allocation subject to per-user rate and BER requirements.
Chow, Cioffi, and Bingham~\cite{chow1995dmt} developed a practical bit-loading algorithm for discrete multitone (DMT) systems (the precursor to OFDM in DSL) that minimizes power subject to a target bit rate per subchannel.
Their approach, and the large body of DSL/DMT bit-loading work that followed, treats per-channel targets as hard constraints: the allocator either meets each target exactly or declares infeasibility.
Surveys by Sadr~et~al.~\cite{sadr2009survey} and Huang and Berry~\cite{huangberry2009chapter} classify OFDM resource allocation into rate-adaptive (maximize rate under power constraints) and margin-adaptive (minimize power under rate constraints) categories.
Our formulation occupies a distinct niche: rather than imposing hard rate constraints, we penalize deviations from targets quadratically, producing a smooth allocation that degrades gracefully when the budget is insufficient to meet all targets.
This soft-target approach avoids the infeasibility issues inherent in hard-constraint formulations and provides a differentiable objective that is amenable to integration with gradient-based outer-loop optimizers.

\subsection{Fairness and Proportional Allocation}

Kelly~\cite{kelly1998rate} introduced proportional fairness for rate allocation in communication networks, connecting it to a logarithmic utility maximization.
Mo and Walrand~\cite{mo2000fair} unified proportional, max-min, and intermediate notions of fairness through a parameterized utility family.
Shen~et~al.~\cite{shen2005proportional} applied proportional rate constraints to multiuser OFDM, ensuring that users receive rates in specified proportions.
Our target-rate formulation can be viewed as complementary: proportional fairness controls \emph{ratios} between user rates, whereas our approach controls \emph{absolute deviations} from specified targets.
Moreover, proportional fairness maximizes a monotone concave utility of rate and therefore always exhausts the power budget; even with elastic rate constraints, it does not directly minimize deviation from specified per-channel targets.

\subsection{Lambert~W Function in Wireless Optimization}

The Lambert~W function, defined by $W(z)e^{W(z)} = z$~\cite{corless1996lambertw}, arises in several wireless resource allocation problems.
Brah~et~al.~\cite{brah2011lambertw} used it for constrained power allocation in cooperative relay networks.
Bacci~et~al.~\cite{bacci2014energyaware} derived Lambert~W solutions for energy-aware link adaptation in small-cell networks.
Meshkati~et~al.~\cite{meshkati2006gametheoretic} obtained Lambert~W expressions in game-theoretic energy-efficient power control.
In each of these works, the Lambert~W function emerges from the interplay between logarithmic rate functions and power constraints in the KKT conditions.
Our contribution extends this family of results to the least-squares target-rate tracking setting.

\subsection{Positioning of This Work}

It is instructive to contrast our formulation with the two dominant paradigms.
In general \emph{utility-maximization} frameworks~\cite{kim2010sparecapacity}, the objective is a concave, monotonically increasing function of rate (e.g., weighted sum rate, $\alpha$-fair utilities), so that more power always increases utility and the budget is always fully spent.
Our quadratic deviation objective is non-monotone: once a channel reaches its target, additional power is wasteful, leading to the fundamentally different property of slack power.
\emph{Penalty-based QoS} approaches augment throughput objectives with constraint-violation penalties, typically enforced through Lagrangian terms; the penalty serves as a regularizer rather than as the primary objective.
In contrast, our formulation \emph{directly} minimizes the sum of squared deviations as the sole objective, leading to a structurally different KKT system that admits a Lambert~W closed form.
Penalty methods append deviation terms to a monotone utility that still exhausts the budget; because our objective is \emph{solely} the deviation, the resulting allocation exhibits non-monotone, slack-power behavior that is absent from penalty-augmented formulations.
Prior Lambert~W results~\cite{brah2011lambertw,bacci2014energyaware,meshkati2006gametheoretic} address energy efficiency and cooperative relays rather than squared-deviation target tracking over parallel channels.

Table~\ref{tab:positioning} summarizes how our formulation relates to the classical alternatives.
The target-rate least-squares approach occupies the previously unexplored intersection of soft QoS objectives and closed-form Lambert~W solutions for parallel channels.

\begin{table}[t]
\centering
\caption{Comparison of power allocation design philosophies.}
\label{tab:positioning}
\begin{tabular}{lll}
\toprule
Approach & Objective & Power usage \\
\midrule
Waterfilling & Max sum rate & Always full \\
Margin-adaptive & Min power for hard targets & Minimum needed \\
Max-min fairness & Max worst-case rate & Always full \\
\textbf{Target-rate (ours)} & \textbf{Min squared deviation} & \textbf{May be slack} \\
\bottomrule
\end{tabular}
\end{table}

\section{System Model and Problem Formulation}
\label{sec:model}

\subsection{Parallel Gaussian Channel Model}

We consider a single-user system with $N$ parallel Gaussian channels, as arises in OFDM transmission over a frequency-selective fading channel.
Let $P_i \ge 0$ denote the transmit power allocated to channel~$i$, and let $a_i > 0$ denote the effective channel gain-to-noise ratio (incorporating path loss, fading, and noise power density) on channel~$i$.
We assume perfect channel state information (CSI) at the transmitter, so that the $a_i$ are known.
The achievable spectral efficiency\footnote{Throughout this paper, ``rate'' refers to spectral efficiency (bits/s/Hz), i.e., the rate normalized by bandwidth.} on channel~$i$ is given by the Shannon formula
\begin{equation}
r_i(P_i) = \log_2(1 + a_i P_i) \quad \text{(bits/s/Hz)}.
\label{eq:rate}
\end{equation}
This model applies directly to OFDM subcarriers, DMT tones in DSL, or any set of non-interfering parallel links.

\subsection{Optimization Problem}

Each channel~$i$ is assigned a target spectral efficiency $T_i \ge 0$.
We denote the power allocation vector by $\bm{P} = (P_1, \ldots, P_N)$.
Given a total power budget $P_{\mathrm{tot}} > 0$, we seek the allocation that minimizes the aggregate squared deviation from these targets:
\begin{equation}
\begin{aligned}
\min_{\{P_i\}_{i=1}^N}\quad & J(\bm{P}) \triangleq \sum_{i=1}^{N}\Big(\log_2(1+a_iP_i)-T_i\Big)^2 \\
\text{s.t.}\quad & \sum_{i=1}^{N} P_i \le P_{\mathrm{tot}}, \qquad P_i\ge 0\;\;\forall\, i.
\end{aligned}
\label{eq:main}
\end{equation}
The targets $T_i$ may be heterogeneous across channels.
A weighted variant $\sum_i w_i(r_i - T_i)^2$ with importance weights $w_i > 0$ can be treated by the same approach (see Appendix).

\subsection{Discussion of the Objective Structure}

The rate function $r_i(P_i) = \log_2(1 + a_i P_i)$ is concave in $P_i$.
Let $c \triangleq \ln 2$ (used throughout Sections~\ref{sec:model}--\ref{sec:kkt}).
Since each term $f_i(P_i) = (r_i(P_i) - T_i)^2$ composes a convex quadratic with a concave function, convexity of the objective is \emph{not immediate}: the second derivative
\begin{equation}
f_i''(P_i) = \frac{2a_i^2}{c^2(1+a_iP_i)^2}\Big[1 - \ln(1+a_iP_i) + T_i\, c\Big]
\label{eq:secondderiv}
\end{equation}
changes sign when $\ln(1 + a_i P_i) > 1 + T_i\, c$, i.e., when $P_i > (e \cdot 2^{T_i} - 1)/a_i$.
Consequently, $J(\bm{P})$ is \emph{not} globally convex over $P_i \ge 0$.
However, as we establish in Section~\ref{sec:properties}, the optimal solution is confined to a restricted domain on which convexity holds, enabling an exact closed-form solution.

\section{Structural Properties}
\label{sec:properties}

We establish structural results that distinguish the target-rate allocation from classical waterfilling and then use them to prove convexity on the natural domain of the problem.

\begin{theorem}[No Overshoot]
\label{thm:noovershoot}
If $\bm{P}^* = (P_1^*,\ldots,P_N^*)$ is optimal for~\eqref{eq:main}, then $r_i(P_i^*) \le T_i$ for all $i = 1,\ldots,N$.
\end{theorem}

\begin{proof}
Suppose for contradiction that $r_j(P_j^*) > T_j$ for some channel~$j$.
Consider the modified allocation $P_j' = (2^{T_j} - 1)/a_j < P_j^*$, which achieves $r_j(P_j') = T_j$.
This modification strictly reduces the $j$-th term from $(r_j(P_j^*) - T_j)^2 > 0$ to zero, does not increase any other term, and strictly reduces the total power $\sum_i P_i$.
Since the resulting allocation remains feasible, we have $J(\bm{P}') < J(\bm{P}^*)$, contradicting the optimality of~$\bm{P}^*$.
\end{proof}

Theorem~\ref{thm:noovershoot} implies that the optimal allocation is always in the region where each $r_i \le T_i$, which, as we show next, is precisely the region where the objective is convex.
From a practical standpoint, the no-overshoot property means the system never wastes power pushing a channel beyond its target; excess power on one channel would be better spent reducing the shortfall on another.

\begin{proposition}[Per-Channel Cap]
\label{prop:cap}
The power required to exactly meet the target on channel~$i$ is
\begin{equation}
\bar{P}_i \triangleq \frac{2^{T_i} - 1}{a_i}.
\label{eq:cap}
\end{equation}
By Theorem~\ref{thm:noovershoot}, the optimal allocation satisfies $P_i^* \le \bar{P}_i$ for all~$i$.
\end{proposition}

\subsection{Convexity on the Restricted Domain}

Theorem~\ref{thm:noovershoot} and Proposition~\ref{prop:cap} show that optimal solutions satisfy $0 \le P_i^* \le \bar{P}_i$.
We may therefore add the constraints $P_i \le \bar{P}_i$ to~\eqref{eq:main} without altering the optimal solution.
We now verify that the objective is convex over this restricted domain.

\begin{proposition}[Convexity]
\label{prop:convexity}
Each term $f_i(P_i) = (\log_2(1 + a_i P_i) - T_i)^2$ satisfies $f_i''(P_i) \ge 0$ for all $P_i \in [0, \bar{P}_i]$.
\end{proposition}

\begin{proof}
From~\eqref{eq:secondderiv}, $f_i''(P_i) \ge 0$ if and only if $\ln(1 + a_i P_i) \le 1 + T_i \ln 2$.
The inflection point occurs at $P_i^{\mathrm{inf}} = (e \cdot 2^{T_i} - 1)/a_i$.
Since $e > 2$, we have $e \cdot 2^{T_i} > 2^{T_i+1} > 2^{T_i}$ for all $T_i \ge 0$, and therefore $P_i^{\mathrm{inf}} > \bar{P}_i$.
It follows that $f_i''(P_i) \ge 0$ for all $P_i \in [0, \bar{P}_i]$.
\end{proof}

Since $J(\bm{P})$ is a sum of convex functions over the compact domain $\{0 \le P_i \le \bar{P}_i,\; \sum_i P_i \le P_{\mathrm{tot}}\}$ with linear constraints, the augmented problem is a convex program.
By Slater's condition (the origin $\bm{P} = \bm{0}$ is strictly feasible), strong duality holds and the KKT conditions are both necessary and sufficient for global optimality~\cite{boyd2004convex}.
Furthermore, because the additional constraints $P_i \le \bar{P}_i$ are inactive at the optimum (by Theorem~\ref{thm:noovershoot}, the optimal rates satisfy $r_i^* \le T_i$ but generically $r_i^* < T_i$ in Case~B), the dual variables associated with these upper bounds are zero, and the KKT conditions for the original problem~\eqref{eq:main} coincide with those of the augmented problem.

\begin{theorem}[Regime Characterization]
\label{thm:regimes}
The optimal solution exhibits two distinct regimes:
\begin{itemize}
\item \textbf{Case~A} (sufficient power): If $P_{\mathrm{tot}} \ge \sum_{i=1}^{N} \bar{P}_i$, then $P_i^* = \bar{P}_i$ for all~$i$, the objective attains $J^* = 0$, and the power constraint is \emph{slack} with $\sum_i P_i^* = \sum_i \bar{P}_i < P_{\mathrm{tot}}$.
\item \textbf{Case~B} (insufficient power): If $P_{\mathrm{tot}} < \sum_{i=1}^{N} \bar{P}_i$, then $\sum_i P_i^* = P_{\mathrm{tot}}$ (the power constraint is tight), $J^* > 0$, and the optimal dual variable $\lambda^* > 0$.
\end{itemize}
\end{theorem}

\begin{proof}
For Case~A, setting $P_i = \bar{P}_i$ achieves $r_i = T_i$ for all~$i$, giving $J = 0$.
Since $J \ge 0$ by definition, this is globally optimal.
The total power used is $\sum_i \bar{P}_i \le P_{\mathrm{tot}}$, so the constraint is slack.

For Case~B, if $\sum_i P_i^* < P_{\mathrm{tot}}$, then $\lambda^* = 0$ by complementary slackness.
With $\lambda^* = 0$, the unconstrained optimum of each term $(r_i - T_i)^2$ is $P_i = \bar{P}_i$, giving total power $\sum_i \bar{P}_i > P_{\mathrm{tot}}$, which is a contradiction.
Hence the constraint must be tight.
\end{proof}

\begin{remark}
The slack-power regime (Case~A) is a distinctive feature of target-rate allocation.
In contrast, classical waterfilling for sum-rate maximization always uses the entire budget, since additional power always increases the sum rate.
\end{remark}

\begin{corollary}[Monotonicity]
\label{cor:monotone}
The optimal objective $J^*(P_{\mathrm{tot}})$ is nonincreasing in $P_{\mathrm{tot}}$, reaching zero at $P_{\mathrm{tot}} = \sum_i \bar{P}_i$ and remaining zero for all larger budgets.
Similarly, for fixed $P_{\mathrm{tot}}$, $J^*$ is nondecreasing in each target $T_i$.
\end{corollary}

\section{KKT Solution via Lambert~W}
\label{sec:kkt}

In this section, we derive the closed-form per-channel power allocation using the KKT optimality conditions and the Lambert~W function.

\subsection{KKT Conditions}

Let $c = \ln 2$.
We form the Lagrangian
\begin{equation}
\mathcal{L}(\bm{P},\lambda,\bm{\mu}) = \sum_{i=1}^{N}(r_i - T_i)^2 + \lambda\Big(\sum_{i=1}^{N} P_i - P_{\mathrm{tot}}\Big) - \sum_{i=1}^{N} \mu_i P_i
\end{equation}
with dual variables $\lambda \ge 0$ (power constraint) and $\mu_i \ge 0$ (nonnegativity).
As established in Section~\ref{sec:properties}, the problem is convex on the restricted domain $0 \le P_i \le \bar{P}_i$ and Slater's condition holds, so the KKT conditions are necessary and sufficient for optimality~\cite{boyd2004convex}.
Stationarity with respect to $P_i$ yields
\begin{equation}
\frac{\partial \mathcal{L}}{\partial P_i} = 2(r_i - T_i)\frac{a_i}{(1+a_iP_i)\,c} + \lambda - \mu_i = 0.
\label{eq:stationary}
\end{equation}
The complementary slackness conditions are $\lambda(\sum_i P_i - P_{\mathrm{tot}}) = 0$ and $\mu_i P_i = 0$ for all~$i$.

\subsection{Active-Set Derivation}

For an active channel ($P_i > 0$), we have $\mu_i = 0$ and~\eqref{eq:stationary} becomes
\begin{equation}
2\Big(\log_2(1+a_iP_i) - T_i\Big)\frac{a_i}{(1+a_iP_i)\,c} + \lambda = 0.
\label{eq:active}
\end{equation}
Since $r_i \le T_i$ by Theorem~\ref{thm:noovershoot}, the first term is nonpositive, which is consistent with $\lambda \ge 0$.
Substituting $y_i = 1 + a_i P_i > 1$ and using $\log_2 y_i = \ln y_i / c$, equation~\eqref{eq:active} becomes
\begin{equation}
\ln y_i = T_i\, c - \frac{\lambda\, c^2}{2\, a_i}\, y_i.
\label{eq:lny}
\end{equation}
Exponentiating both sides yields
\begin{equation}
y_i \, e^{\frac{\lambda c^2}{2a_i}\, y_i} = 2^{T_i}.
\label{eq:yexp}
\end{equation}
We define $\alpha_i \triangleq \lambda c^2 / (2a_i) \ge 0$ and multiply~\eqref{eq:yexp} by $\alpha_i$:
\begin{equation}
(\alpha_i y_i)\, e^{\alpha_i y_i} = \alpha_i\, 2^{T_i}.
\end{equation}
This is in the canonical form $z\,e^z = w$ with $z = \alpha_i y_i$ and $w = \alpha_i\, 2^{T_i} \ge 0$.
Since $w \ge 0$, the principal branch $W_0$ of the Lambert~W function~\cite{corless1996lambertw}, the unique real-valued branch satisfying $W_0(w) \ge -1$ for $w \ge -1/e$, yields $\alpha_i y_i = W_0(\alpha_i\, 2^{T_i})$, and hence
\begin{equation}
\boxed{
P_i(\lambda) = \frac{2}{\lambda\, c^2}\, W_0\!\left(\frac{\lambda\, c^2}{2\, a_i}\, 2^{T_i}\right) - \frac{1}{a_i}.
}
\label{eq:lambert}
\end{equation}
This is the closed-form power allocation for each active channel as a function of the dual variable~$\lambda$.
From a practical standpoint, equation~\eqref{eq:lambert} provides a direct mapping from channel parameters $(a_i, T_i)$ and the single scalar~$\lambda$ to the optimal per-channel power, requiring only one Lambert~W evaluation per channel, and no iterative per-channel optimization is needed.
Note that $c^2 = (\ln 2)^2$ appears because we use $\log_2$ in the rate expression; for natural logarithm rates, this simplifies to $c = 1$.

\subsection{Inactivity Threshold}

When $P_i = 0$, we have $y_i = 1$ and $r_i = 0$.
The stationarity condition~\eqref{eq:stationary} gives $\mu_i = \lambda - 2a_i T_i / c$.
Dual feasibility $\mu_i \ge 0$ yields the inactivity condition:
\begin{equation}
P_i^* = 0 \quad\text{if}\quad \lambda \ge \frac{2\,a_i\,T_i}{c}.
\label{eq:threshold}
\end{equation}
Channels with small $a_i T_i$ products (poor gain and/or low target) are shut off first as $\lambda$ increases.
Combining~\eqref{eq:lambert} and~\eqref{eq:threshold} with the cap from Proposition~\ref{prop:cap}:
\begin{equation}
P_i^*(\lambda) = \min\Big\{\bar{P}_i,\;\max\big\{0,\; P_i(\lambda)\big\}\Big\}.
\label{eq:final}
\end{equation}
This expression clips each channel's allocation to the feasible range $[0, \bar{P}_i]$: channels with sufficient power reach their caps, channels with limited power receive the Lambert~W allocation, and channels below the activity threshold~\eqref{eq:threshold} are shut off entirely.
The upper clamp $\bar{P}_i$ enforces the no-overshoot property (Theorem~\ref{thm:noovershoot}), ensuring no channel receives more power than needed to meet its target.

\subsection{Monotonicity and Limiting Behavior}

\begin{lemma}[Monotonicity of $P_i(\lambda)$]
\label{lem:monotone}
For each channel~$i$ with $a_i, T_i > 0$, the unclamped allocation $P_i(\lambda)$ in~\eqref{eq:lambert} is strictly decreasing in $\lambda$ for $\lambda > 0$.
\end{lemma}

\begin{proof}
Let $w = \alpha_i\, 2^{T_i}$ where $\alpha_i = \lambda c^2 / (2a_i)$.
Then $P_i(\lambda) = W_0(w)/\alpha_i - 1/a_i$.
Using the identity $W_0'(w) = W_0(w)/[w(1 + W_0(w))]$, a direct computation shows $dP_i/d\lambda < 0$ for $\lambda > 0$, since both $W_0(w)$ and $\alpha_i$ increase with $\lambda$ but $P_i$ decreases.
\end{proof}

\begin{lemma}[Limiting Behavior]
\label{lem:limits}
As $\lambda \to 0^+$, $P_i(\lambda) \to \bar{P}_i = (2^{T_i}-1)/a_i$ (the cap).
As $\lambda \to \infty$, the unclamped expression~\eqref{eq:lambert} tends to $-1/a_i$, so that the clamped allocation $P_i^*(\lambda) \to 0$.
\end{lemma}

\begin{proof}
As $\lambda \to 0^+$, the argument $w = \alpha_i 2^{T_i} \to 0$.
Using $W_0(w) \approx w$ for small $w$, we get $P_i \to 2^{T_i}/a_i - 1/a_i = (2^{T_i}-1)/a_i = \bar{P}_i$.
As $\lambda \to \infty$, $w \to \infty$ and $W_0(w) \sim \ln w$, so $W_0(w)/w \to 0$.
Then $y_i = W_0(w)/\alpha_i \to 0$, and the unclamped $P_i(\lambda) = (y_i - 1)/a_i \to -1/a_i$.
Since $P_i^*(\lambda) = \max\{0, P_i(\lambda)\}$, the clamped allocation tends to zero.
\end{proof}

\subsection{Special Cases}

\textbf{Single channel ($N=1$):}
If $P_{\mathrm{tot}} \ge \bar{P}_1$, then $P_1^* = \bar{P}_1$ and $J^* = 0$.
Otherwise, $P_1^* = P_{\mathrm{tot}}$ (use all available power on the single channel).

\textbf{Equal gains ($a_i = a$ for all $i$):}
With uniform targets $T_i = T$, symmetry gives $P_i^* = \min(P_{\mathrm{tot}}/N,\, \bar{P})$ for all~$i$, where $\bar{P} = (2^T - 1)/a$.
The allocation is uniform, contrasting with waterfilling (also uniform in this case but using all power).

\textbf{Target $T_i \to 0$:}
As $T_i \to 0$, the cap $\bar{P}_i \to 0$ and channel~$i$ receives vanishing power.

\textbf{Target $T_i \to \infty$:}
As $T_i \to \infty$, the cap $\bar{P}_i \to \infty$ and the channel demands arbitrarily large power.
In practice, the sum-power constraint limits the achievable rate.

\section{Algorithm}
\label{sec:algorithm}

In this section, we present the bisection algorithm for finding the optimal dual variable and analyze its convergence and numerical properties.

\subsection{Bisection Procedure}

By Lemma~\ref{lem:monotone}, the total allocated power $S(\lambda) \triangleq \sum_{i=1}^{N} P_i^*(\lambda)$ is nonincreasing in~$\lambda$.
By Lemma~\ref{lem:limits}, $S(0^+) = \sum_i \bar{P}_i$ and $S(\infty) = 0$.
In Case~B ($P_{\mathrm{tot}} < \sum_i \bar{P}_i$), the optimal $\lambda^*$ is the unique root of $S(\lambda^*) = P_{\mathrm{tot}}$, which can be found via bisection on~$\lambda$.
Algorithm~\ref{alg:bisection} presents the complete procedure.

\begin{algorithm}[t]
\caption{Target-rate least-squares power allocation via bisection}
\label{alg:bisection}
\begin{algorithmic}[1]
\Require Channel gains $\{a_i\}_{i=1}^N$, targets $\{T_i\}_{i=1}^N$, budget $P_{\mathrm{tot}}$, tolerance $\varepsilon$
\State Compute caps $\bar{P}_i \leftarrow (2^{T_i}-1)/a_i$ for all $i$
\If{$P_{\mathrm{tot}} \ge \sum_i \bar{P}_i$}
    \State \Return $P_i^* = \bar{P}_i$ for all $i$; \quad $\lambda^* = 0$
\EndIf
\State $\lambda_{\mathrm{lo}} \leftarrow 0$
\State $\lambda_{\mathrm{hi}} \leftarrow 1$
\While{$S(\lambda_{\mathrm{hi}}) > P_{\mathrm{tot}}$}
    \State $\lambda_{\mathrm{hi}} \leftarrow 2\,\lambda_{\mathrm{hi}}$
\EndWhile
\Repeat
    \State $\lambda \leftarrow (\lambda_{\mathrm{lo}} + \lambda_{\mathrm{hi}})/2$
    \For{$i = 1$ \textbf{to} $N$}
        \State Compute $P_i^*(\lambda)$ via~\eqref{eq:lambert} and~\eqref{eq:final}
    \EndFor
    \State $S(\lambda) \leftarrow \sum_{i=1}^{N} P_i^*(\lambda)$
    \If{$S(\lambda) > P_{\mathrm{tot}}$}
        \State $\lambda_{\mathrm{lo}} \leftarrow \lambda$
    \Else
        \State $\lambda_{\mathrm{hi}} \leftarrow \lambda$
    \EndIf
\Until{$|S(\lambda) - P_{\mathrm{tot}}| < \varepsilon$}
\State \Return $\{P_i^*(\lambda)\}_{i=1}^N$ and $\lambda$
\end{algorithmic}
\end{algorithm}

\subsection{Convergence Analysis}

\begin{theorem}[Convergence]
\label{thm:convergence}
Under Case~B ($P_{\mathrm{tot}} < \sum_i \bar{P}_i$) with at least one channel satisfying $a_i, T_i > 0$, Algorithm~\ref{alg:bisection} terminates in at most $O(\log(1/\varepsilon))$ bisection iterations.
Each iteration evaluates $N$ Lambert~W function calls and elementary operations, giving a total complexity of $O(N \log(1/\varepsilon))$.
\end{theorem}

\begin{proof}
The doubling loop (lines~6--8) finds an upper bound $\lambda_{\mathrm{hi}}$ in at most $O(\log(\lambda^*))$ steps.
Each bisection iteration halves the interval $[\lambda_{\mathrm{lo}}, \lambda_{\mathrm{hi}}]$.
After $m$ iterations, the interval width is $\lambda_{\mathrm{hi}}^{(0)} / 2^m$.
Since $S(\lambda)$ is Lipschitz continuous near $\lambda^*$, the error $|S(\lambda) - P_{\mathrm{tot}}|$ decreases proportionally, reaching $\varepsilon$ in $O(\log(1/\varepsilon))$ iterations.
Each iteration computes~\eqref{eq:lambert} for all $N$ channels in $O(N)$ time.
\end{proof}

\subsection{Numerical Stability}

For very small $\lambda$, the expression~\eqref{eq:lambert} involves the ratio $W_0(w)/\alpha_i$ where both $W_0(w) \to 0$ and $\alpha_i \to 0$.
Using the asymptotic expansion $W_0(w) \approx w - w^2 + \cdots$ for small $w$, we obtain
\begin{equation}
P_i(\lambda) \approx \frac{2^{T_i} - 1}{a_i} - \frac{\lambda\, c^2\, 2^{2T_i}}{2\,a_i^2} + O(\lambda^2),
\end{equation}
confirming the smooth transition to the cap solution as $\lambda \to 0$.
In practice, we handle $\lambda \approx 0$ by directly checking whether $P_{\mathrm{tot}} \ge \sum_i \bar{P}_i$ (line~2 of Algorithm~\ref{alg:bisection}) and returning the cap solution, bypassing the Lambert~W evaluation entirely.

\subsection{Inactive Channel Pruning}

By the inactivity threshold~\eqref{eq:threshold}, any channel satisfying $\lambda \ge 2a_i T_i / c$ receives zero power.
During the bisection, once $\lambda_{\mathrm{lo}}$ exceeds the threshold for channel~$i$, that channel can be removed from further Lambert~W evaluations for all subsequent iterations, reducing the per-iteration cost.
In practice, channels with small $a_i T_i$ products are pruned early, and the effective per-iteration cost decreases as the bisection progresses.

\subsection{Warm-Starting for Time-Varying Channels}

In adaptive OFDM systems where channel gains change over time, the optimal $\lambda^*$ at time~$t$ is typically close to $\lambda^*$ at time~$t-1$.
We can warm-start the bisection by initializing $[\lambda_{\mathrm{lo}}, \lambda_{\mathrm{hi}}]$ around the previous solution, significantly reducing the number of iterations needed when the channel varies slowly.

\section{Numerical Results}
\label{sec:results}

We validate the analytical results and compare target-rate allocation with classical waterfilling through numerical experiments.
All simulations are implemented in Python~3.11.
The Lambert~W evaluations use \texttt{scipy.special.lambertw} from SciPy~1.12~\cite{scipy_lambertw}, with bisection tolerance $\varepsilon = 10^{-10}$.
The numerical baseline uses the Sequential Least-Squares Programming (SLSQP) solver from \texttt{scipy.optimize.minimize} with an analytical gradient and constraint tolerance $10^{-12}$.
Timing measurements are averaged over five independent runs on a single core of an AMD Ryzen~9 5900X processor and exclude initial import overhead.

\subsection{Setup}

Unless otherwise stated, we use $N = 8$ channels with deterministic gains $a_i \in \{20, 15, 10, 7, 5, 3, 2, 1\}$ (modeling a mix of strong and weak subcarriers), uniform target $T = 3$~bits/s/Hz, and total power budget $P_{\mathrm{tot}} = 10$.
For statistical experiments, we use Rayleigh fading with mean SNR of 10~dB.

\subsection{Validation Against Numerical Optimization}

Fig.~\ref{fig:validation} compares the objective value obtained by the Lambert~W closed form (Algorithm~\ref{alg:bisection}) with the result from a general-purpose SLSQP nonlinear optimizer, sweeping $P_{\mathrm{tot}}$ from 0.5 to 25.
The two curves overlap to plotting accuracy across the entire range, confirming the correctness of the closed-form solution.
The objective decreases monotonically with $P_{\mathrm{tot}}$ and reaches zero at the threshold $\sum_i \bar{P}_i \approx 16.8$.

\begin{figure}[t]
\centering
\includegraphics[width=0.85\textwidth]{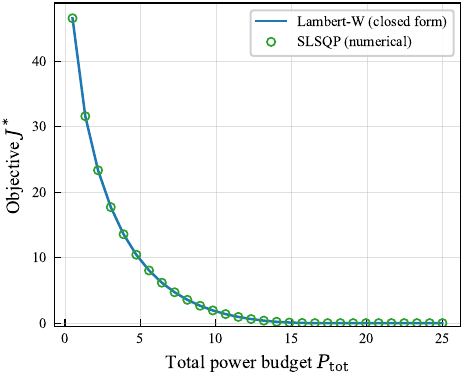}
\caption{Objective $J^*$ vs.\ total power budget $P_{\mathrm{tot}}$, comparing the Lambert~W closed form with SLSQP numerical optimization.
The two solutions agree to machine precision.}
\label{fig:validation}
\end{figure}

\subsection{Allocation Comparison}

Fig.~\ref{fig:allocation} shows the per-channel power allocation for target-rate and waterfilling approaches.
Waterfilling favors the strongest channels and always uses the full budget.
In contrast, the target-rate allocation distributes power to bring each channel as close to the target as possible: weaker channels receive proportionally more power relative to their caps, and no channel exceeds its cap~$\bar{P}_i$.

\begin{figure}[t]
\centering
\includegraphics[width=0.85\textwidth]{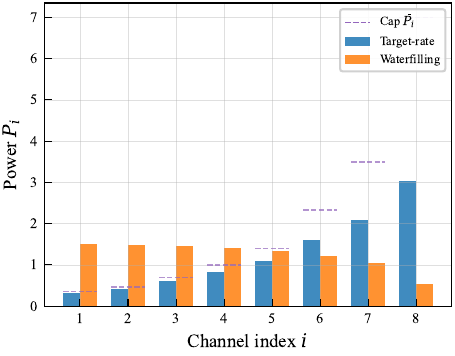}
\caption{Per-channel power allocation for $N = 8$ channels.
The target-rate allocation respects per-channel caps (dashed lines) and distributes power to minimize rate deviations.
Waterfilling concentrates power on the strongest channels.}
\label{fig:allocation}
\end{figure}

Fig.~\ref{fig:rates} illustrates the no-overshoot property (Theorem~\ref{thm:noovershoot}).
Under target-rate allocation, every achieved rate satisfies $r_i \le T$ (the dashed target line).
In contrast, waterfilling pushes strong channels well above the target while leaving weak channels far below.

\begin{figure}[t]
\centering
\includegraphics[width=0.85\textwidth]{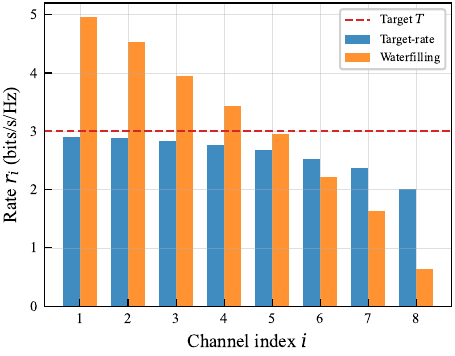}
\caption{Achieved rates per channel.
Target-rate allocation never exceeds the target $T = 3$ (no overshoot).
Waterfilling overshoots on strong channels and falls short on weak ones.}
\label{fig:rates}
\end{figure}

\subsection{Dual Variable Search}

Fig.~\ref{fig:dual} visualizes the monotonically decreasing function $S(\lambda)$.
The optimal $\lambda^*$ is found at the intersection with $P_{\mathrm{tot}} = 10$, confirming the bisection approach of Algorithm~\ref{alg:bisection}.

\begin{figure}[t]
\centering
\includegraphics[width=0.85\textwidth]{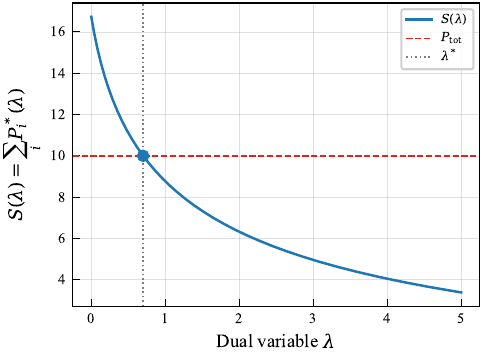}
\caption{Total power $S(\lambda) = \sum_i P_i^*(\lambda)$ as a function of the dual variable~$\lambda$.
The optimal $\lambda^*$ satisfies $S(\lambda^*) = P_{\mathrm{tot}}$.
Monotonicity guarantees unique convergence via bisection.}
\label{fig:dual}
\end{figure}

\subsection{Objective Landscape and Phase Transition}

Fig.~\ref{fig:landscape} shows how the optimal objective $J^*$ varies with the power budget.
A clear phase transition occurs at $P_{\mathrm{tot}} = \sum_i \bar{P}_i \approx 16.8$: below this threshold, the objective decreases steeply as more power becomes available; above it, $J^* = 0$ and additional power goes unused.
This illustrates the regime characterization of Theorem~\ref{thm:regimes}.

\begin{figure}[t]
\centering
\includegraphics[width=0.85\textwidth]{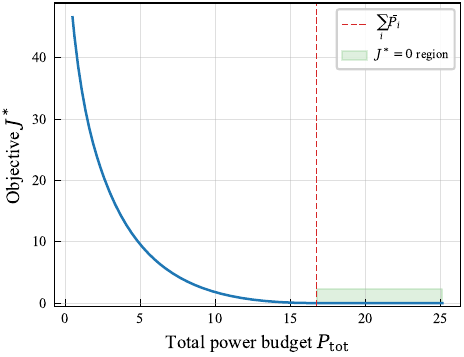}
\caption{Optimal objective vs.\ power budget.
The objective reaches zero at $\sum_i \bar{P}_i \approx 16.8$ (dashed line), beyond which power goes unused.}
\label{fig:landscape}
\end{figure}

\subsection{Heterogeneous Targets}

Fig.~\ref{fig:hetero} demonstrates allocation with heterogeneous targets $T_i \in \{5, 4, 3, 3, 2, 2, 1, 1\}$ under two power budgets that illustrate the two regimes of Theorem~\ref{thm:regimes}.
With these targets and gains, the sum of per-channel caps is $\sum_i \bar{P}_i \approx 7.35$.

The left column shows $P_{\mathrm{tot}} = 5$ (Case~B: tight budget).
Since $5 < 7.35$, the power constraint is active and the entire budget is consumed.
Channels with higher targets receive more power, but weaker channels with high targets show the largest rate shortfalls, a natural consequence of the least-squares objective.
This pattern arises because weak channels require exponentially more power per unit of rate; when the budget is tight, the objective preferentially reduces deviations where the marginal improvement is greatest.

The right column shows $P_{\mathrm{tot}} = 15$ (Case~A: surplus budget).
Since $15 > 7.35$, the allocator sets each $P_i^* = \bar{P}_i$, achieving $r_i = T_i$ for every channel with $J^* = 0$.
The total power used is only 7.35 out of 15, leaving 7.65 unused, a direct illustration of the slack-power property that distinguishes target-rate allocation from waterfilling.

\begin{figure}[t]
\centering
\includegraphics[width=\textwidth]{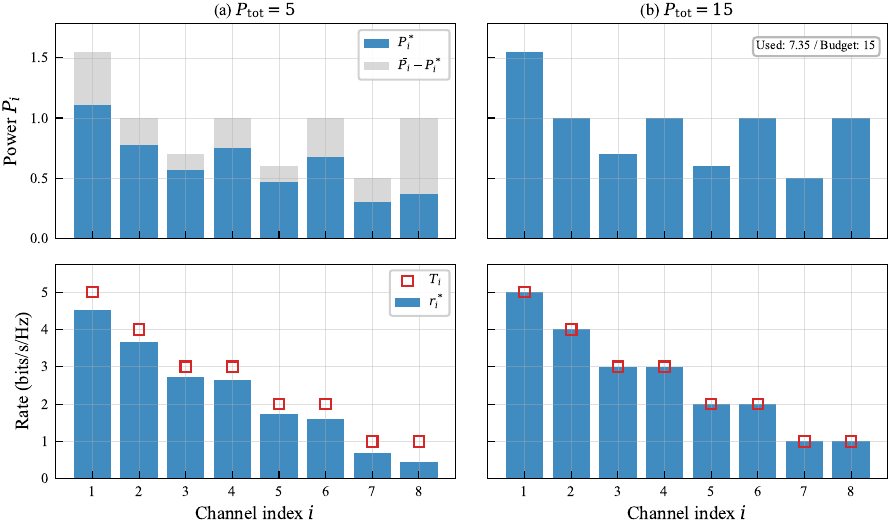}
\caption{Heterogeneous targets $T_i \in \{5,4,3,3,2,2,1,1\}$.
Left ($P_{\mathrm{tot}} = 5$, Case~B): tight budget forces shortfalls on weak, high-target channels.
Right ($P_{\mathrm{tot}} = 15$, Case~A): all targets are met and 7.65 out of 15 power units go unused.
Top row: power allocation (blue) and gap to cap (gray).
Bottom row: achieved rates (blue bars) vs.\ targets (red squares).}
\label{fig:hetero}
\end{figure}

\subsection{Statistical Performance over Fading Channels}

Fig.~\ref{fig:cdf} presents the empirical CDF of per-channel rate deviations $|r_i - T|$ over 1{,}000 independent Rayleigh fading realizations for four allocation strategies: target-rate, waterfilling, uniform, and proportional fairness.
The target-rate allocation consistently achieves smaller deviations: its CDF rises steeply near zero, indicating that most channels are close to the target.
Quantitatively, the target-rate allocation achieves a median deviation of 0.16~bits/s/Hz with a 90th percentile of 1.11, compared to 1.18 (median) and 2.73 (90th percentile) for waterfilling, 1.00 and 2.27 for uniform allocation, and 0.80 and 1.89 for proportional fairness.
While proportional fairness outperforms waterfilling and uniform allocation, it still exhibits substantially larger deviations than the target-rate approach because it optimizes for rate ratios rather than absolute target tracking.

\begin{figure}[t]
\centering
\includegraphics[width=0.85\textwidth]{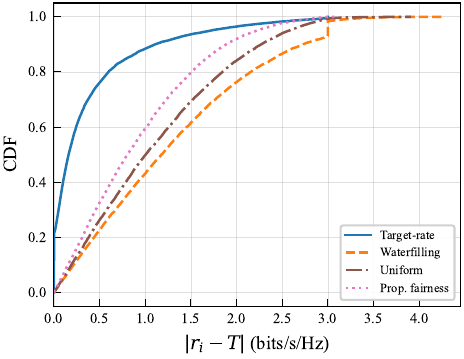}
\caption{Empirical CDF of $|r_i - T|$ over 1{,}000 Rayleigh fading realizations ($N = 8$, $T = 3$, $P_{\mathrm{tot}} = 10$) for four allocation strategies.
Target-rate allocation concentrates deviations near zero.}
\label{fig:cdf}
\end{figure}

\subsection{Computational Scalability}

Fig.~\ref{fig:timing} and Table~\ref{tab:computation} compare the computation time of the Lambert~W bisection algorithm with a general-purpose SLSQP solver for $N$ up to 1{,}024.
The closed-form approach scales nearly linearly in~$N$, consistent with the $O(N\log(1/\varepsilon))$ complexity of Theorem~\ref{thm:convergence}.
At $N = 1{,}024$, the Lambert~W algorithm is approximately 1{,}890$\times$ faster than SLSQP (0.011\,s vs.\ 21.2\,s), making it suitable for real-time OFDM systems with large numbers of subcarriers.
Table~\ref{tab:computation} also reports standard deviations across runs, confirming low timing variability.

\begin{figure}[t]
\centering
\includegraphics[width=0.85\textwidth]{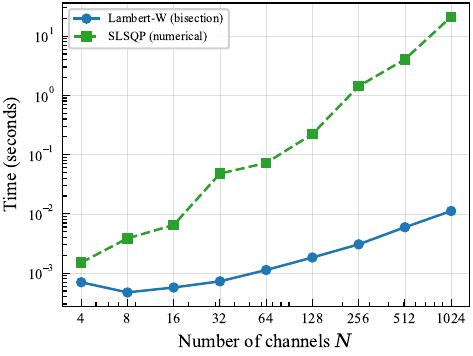}
\caption{Computation time vs.\ number of channels on a log-log scale.
The Lambert~W bisection scales nearly linearly, achieving up to 1{,}890$\times$ speedup over SLSQP at $N = 1{,}024$.}
\label{fig:timing}
\end{figure}

\begin{table}[t]
\centering
\caption{Power usage and target-rate objective $J$ for four allocation strategies. Target-rate allocation leaves power unused when $P_{\mathrm{tot}} \ge \sum_i \bar{P}_i = 16.8$.}
\label{tab:power_usage}
\begin{tabular}{rrrrrrrrr}
\toprule
$P_{\mathrm{tot}}$ & \multicolumn{2}{c}{Target-Rate} & \multicolumn{2}{c}{Waterfilling} & \multicolumn{2}{c}{Uniform} & \multicolumn{2}{c}{Prop.\ Fairness} \\
\cmidrule(lr){2-3} \cmidrule(lr){4-5} \cmidrule(lr){6-7} \cmidrule(lr){8-9}
 & Used & $J$ & Used & $J$ & Used & $J$ & Used & $J$ \\
\midrule
5 & 5.00 & 9.593 & 5.00 & 18.747 & 5.00 & 12.790 & 5.00 & 10.646 \\
10 & 10.00 & 1.789 & 10.00 & 15.383 & 10.00 & 10.600 & 10.00 & 6.532 \\
15 & 15.00 & 0.079 & 15.00 & 17.605 & 15.00 & 13.946 & 15.00 & 8.984 \\
20 & 16.75 & 0.000 & 20.00 & 21.751 & 20.00 & 18.807 & 20.00 & 13.362 \\
25 & 16.75 & 0.000 & 25.00 & 26.582 & 25.00 & 24.125 & 25.00 & 18.396 \\
\bottomrule
\end{tabular}
\end{table}

\begin{table}[t]
\centering
\caption{Computation time (seconds) and speedup of Lambert-W bisection vs.\ SLSQP. Mean $\pm$ std over 5 runs.}
\label{tab:computation}
\begin{tabular}{rrrr}
\toprule
$N$ & Lambert-W & SLSQP & Speedup \\
\midrule
4 & 0.0007 $\pm$ 0.0003 & 0.0015 $\pm$ 0.0002 & 2$\times$ \\
8 & 0.0005 $\pm$ 0.0000 & 0.0039 $\pm$ 0.0002 & 8$\times$ \\
16 & 0.0006 $\pm$ 0.0000 & 0.0065 $\pm$ 0.0002 & 11$\times$ \\
32 & 0.0007 $\pm$ 0.0000 & 0.0478 $\pm$ 0.0063 & 66$\times$ \\
64 & 0.0011 $\pm$ 0.0000 & 0.0723 $\pm$ 0.0118 & 64$\times$ \\
128 & 0.0018 $\pm$ 0.0001 & 0.2221 $\pm$ 0.0039 & 121$\times$ \\
256 & 0.0031 $\pm$ 0.0001 & 1.4235 $\pm$ 0.6380 & 466$\times$ \\
512 & 0.0059 $\pm$ 0.0003 & 4.0471 $\pm$ 1.0514 & 681$\times$ \\
1024 & 0.0112 $\pm$ 0.0001 & 21.1669 $\pm$ 1.1776 & 1896$\times$ \\
\bottomrule
\end{tabular}
\end{table}

\subsection{Sensitivity to Operating SNR}

Fig.~\ref{fig:sensitivity} shows the mean objective $\bar{J}$ as a function of mean channel SNR for all four allocation strategies, averaged over 500 Rayleigh fading realizations per SNR point ($N = 8$, $T = 3$, $P_{\mathrm{tot}} = 10$).
The target-rate allocation maintains the lowest deviation across the entire SNR range and approaches zero at high SNR as more channels become individually feasible.
In contrast, waterfilling deviation \emph{grows} with SNR because stronger channels overshoot the target by increasingly large margins; uniform and proportional fairness allocations exhibit similar growth, though at somewhat lower levels.
This confirms that the target-rate objective is fundamentally better matched to target-tracking applications across a wide range of operating conditions.

\begin{figure}[t]
\centering
\includegraphics[width=0.85\textwidth]{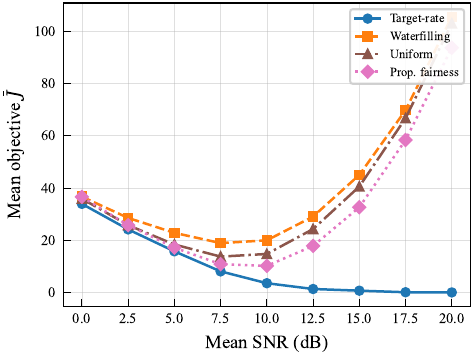}
\caption{Mean objective $\bar{J}$ vs.\ mean SNR over 500 Rayleigh realizations per point ($N = 8$, $T = 3$, $P_{\mathrm{tot}} = 10$).
Target-rate allocation maintains lower deviations across the SNR range.}
\label{fig:sensitivity}
\end{figure}

\subsection{Power Usage Comparison}

Table~\ref{tab:power_usage} quantifies the power usage and target-rate objective $J$ for all four allocation strategies across a range of budgets.
The target-rate allocation achieves strictly lower objective values than all three baselines for every budget level.
Proportional fairness outperforms both waterfilling and uniform allocation but still incurs substantially higher deviations than the target-rate approach.
At $P_{\mathrm{tot}} = 20$ and above, the target-rate allocation leaves power unused after meeting all targets ($J^* = 0$), while all other strategies continue to expend the full budget with increasing rate deviations.

We note that when the design goal is to maximize aggregate throughput rather than track targets, waterfilling remains the optimal strategy.
The target-rate allocation is not a replacement for waterfilling, uniform, or proportional fairness but rather an alternative suited to a different objective.
In the experiments above, waterfilling achieves higher sum rates than all other approaches at every budget level, as expected, since it directly optimizes for throughput.

\section{Discussion}
\label{sec:discussion}

We now discuss connections to related formulations, potential extensions, and limitations of the proposed approach.

\subsection{Connection to Inverse Waterfilling}

The classical waterfilling solution $P_i = (\nu - 1/a_i)_+$ allocates power proportionally to channel quality.
Our target-rate allocation can be viewed as an ``inverse'' philosophy: it allocates power to equalize rate shortfalls rather than to maximize throughput.
The mathematical structure is also inverted: where waterfilling involves a simple linear-plus-threshold formula, our solution requires a nonlinear Lambert~W evaluation, reflecting the more complex relationship between power and quadratic rate deviation.

\subsection{Multi-User Extension}

While we have focused on a single-user setting, the formulation extends naturally to multi-user OFDMA.
In that setting, each user~$u$ has targets $T_{u,i}$ for assigned subcarriers, and the allocator must additionally decide the subcarrier-to-user mapping.
The per-user subproblems retain the structure of~\eqref{eq:main} and admit the same Lambert~W closed form.
Joint subcarrier assignment and power allocation would require combinatorial optimization (e.g., matching or greedy assignment), with the target-rate power allocation as an inner routine.

\subsection{Integration with OFDM Schedulers}

In practical OFDM systems, schedulers assign target rates based on traffic demands, buffer states, and QoS classes.
The target-rate allocation provides a natural interface: the scheduler specifies per-subcarrier targets, and the allocator returns powers that best match them.
The smooth quadratic penalty provides meaningful gradient information for higher-level optimization loops, making the formulation compatible with differentiable scheduling pipelines.

\subsection{Limitations}

Our analysis assumes perfect CSI, flat fading per subchannel, and Gaussian signaling.
In practice, imperfect CSI introduces estimation errors in the $a_i$ coefficients, which could be addressed by robust optimization variants.
The Shannon rate formula serves as an upper bound; practical systems operating with finite block lengths exhibit a rate penalty that grows with decreasing blocklength~\cite{coverandthomas2006}, and the rate function would need adjustment (e.g., via the normal approximation) to capture this effect.
Regarding computational cost, standard implementations evaluate $W_0$ in constant time via Halley's method with cubic convergence~\cite{corless1996lambertw}, making the per-channel overhead negligible even for embedded systems.
We do not address the problem of selecting the targets $T_i$ themselves, which depends on application requirements and higher-layer protocols.

\section{Conclusion}
\label{sec:conclusion}

We have formulated and solved the target-rate least-squares power allocation problem for parallel Gaussian channels.
The optimal allocation exhibits a no-overshoot property, never pushing any channel beyond its target, and may leave power unused when the budget suffices to meet all targets, which is fundamentally different behavior from classical waterfilling.
We have derived a per-channel closed-form solution via the KKT conditions and the Lambert~W function, and we have shown that the optimal dual variable can be found by bisection in $O(N\log(1/\varepsilon))$ time.
Numerical experiments confirm that the closed-form solution matches numerical optimization to machine precision and achieves up to 1{,}890$\times$ speedup over general-purpose solvers at $N = 1{,}024$ subcarriers.

Several directions for future work merit investigation.
First, extending the formulation to multi-user OFDMA with joint subcarrier assignment and target-rate power allocation would broaden the practical applicability.
Second, incorporating robustness to imperfect CSI through stochastic or worst-case formulations would enhance the engineering relevance.
Third, integrating the target-rate allocator with adaptive MCS selection and link adaptation would bridge the gap between Shannon-rate analysis and practical coded systems.
Fourth, exploring connections to rate-distortion theory, where the rate-deviation objective can be interpreted as a distortion measure, may yield additional analytical insights.
Fifth, the differentiable nature of the quadratic objective makes it a natural building block for end-to-end learning pipelines, where target selection and power allocation are jointly optimized via gradient-based methods.
Finally, the closed-form solution is well suited to emerging 6G systems with AI-native resource allocation, where fast analytical solvers can serve as differentiable layers within deep learning architectures.

\appendix
\section{Weighted Variant}
\label{app:weighted}

Consider the weighted objective $J_w(\bm{P}) = \sum_{i=1}^{N} w_i (r_i - T_i)^2$ with $w_i > 0$.
The KKT analysis proceeds identically, with the stationarity condition becoming
\begin{equation}
2\,w_i(r_i - T_i)\frac{a_i}{(1+a_iP_i)\,c} + \lambda - \mu_i = 0.
\end{equation}
Following the same derivation, the active-set solution is
\begin{equation}
P_i(\lambda) = \frac{2\,w_i}{\lambda\, c^2}\, W_0\!\left(\frac{\lambda\, c^2}{2\,w_i\, a_i}\, 2^{T_i}\right) - \frac{1}{a_i},
\end{equation}
and the bisection algorithm applies unchanged.
The weights $w_i$ allow the designer to prioritize certain channels' target tracking over others.

\section*{Acknowledgement}

The author acknowledges the use of Claude (Anthropic) as an AI coding and writing assistant during the preparation of this manuscript.

\bibliographystyle{IEEEtran}
\bibliography{refs}

\end{document}